\documentclass[10pt]{article}

\usepackage{sbc-template}
\usepackage{graphicx,url}
\usepackage[english]{babel}   
\usepackage[utf8]{inputenc}  
\usepackage[pdftex, unicode]{hyperref}
\hypersetup{
    colorlinks=true,
    linkcolor=blue,
    filecolor=magenta,      
    urlcolor=cyan,
    linktocpage=true
  }
\usepackage{bookmark}
\usepackage{rotating}
\usepackage{amssymb}
\usepackage{amsfonts}
\usepackage{amsmath}
\usepackage{amsthm}
\usepackage{mathrsfs}
\usepackage{mathtools}
\usepackage[all]{xy}
\usepackage{graphicx,subfigure}
\usepackage[section]{placeins} %Con este paquete se ponen las figuras en donde uno indique
\usepackage{float}
\newtheorem{theorem}{Theorem}[section]
\newtheorem{example}[theorem]{Example} 
 
\newtheorem{proposition}[theorem]{Proposition}
\newtheorem{corollary}[theorem]{Corollary}
\newtheorem{definition}[theorem]{Definition}
\newtheorem{lemma}[theorem]{Lemma}
\newtheorem{remark}[theorem]{Remark}
\newtheorem{algorithm}[theorem]{Algorithm}
\def\ali{\noindent}
\def\z{\mathbb{Z}}
\def\c{\mathbb{C}}
     
\title{Liouvillian Solutions of Schr\"odinger Equation with Polynomial
	Potentials using Gr\"obner Basis}

\author{Primitivo Bel\'en Acosta-Hum\'anez\inst{1,2} \& Henock Venegas-G\'omez\inst{3}}

\address{ Instituto Superior de Formaci\'on Docente Salom\'e Ure\~na - ISFODOSU\\ Recinto Emilio Prud'Homme, Santiago de los Caballeros, Rep\'ublica Dominicana
  \nextinstitute 
Facultad de Ciencias B\'asicas y Biom\'edicas -- Universidad Simón Bolívar\\
  Barranquilla, Colombia
\nextinstitute
   Escuela de Matem\'aticas -- Universidad Nacional de Colombia\\
  Medellín, Colombia
  \email{primitivo.acosta-humanez@isfodosu.edu.do,
  	hvenegasg@unal.edu.co}
 }

\begin{document} 

\maketitle

\begin{abstract}
The main aim of this paper is the presentation of a new methodology to obtain Liouvillian solutions of stationary one dimensional Schr\"odinger equation with quasi-solvable polynomial potentials through the using of differential Galois theory and Gr\"obner basis. We illustrate these results by the computing of polynomial potentials of degree 4, 6, 8, 10, 12, 14. Moreover, we show an implementation in \textbf{Mathematica} for the decatic potential.
\end{abstract}

\section*{Introduction}
This paper contains the interaction of four important areas of research: quantum mechanics, algebraic geometry, differential algebra and symbolic computation. In particular, the paper is devoted to the obtaining of explicit eigenvalues of Schr\"odinger operators in non-relativistic quantum mechanics, the so-called solvable and quasi-solvable models in quantum mechanics, using theoretical results of differential Galois theory and computational methods that includes Gr\"obner basis.\medskip

\noindent The obtaining of explicit solutions of the non-relativistic Schr\"odinger Equation, the so-called solvable models, has been a very important topic in where Landau was interested, see \cite{Lan}. Therefore, quasi-solvable models in quantum mechanics is a recent research topic for people working in mathematical physics. The seminal works of Bender and Dunne, see \cite{BeDu,BrNa}, followed by \cite{ChMo,handy,liu}, increased the interest of a large number of researchers in mathematics and physics for the study of quasi-solvable models in quantum mechanics.\medskip

\noindent On the other hand, algebraic techniques and computational methods  helped to the development of differential equations. In particular, the theory developed by Picard and Vessiot, the Galois theory in the context of linear differential equations, known as Differential Galois Theory, provides algorithmic aspects to solve a large number of differential equations. One of these algorithmic  tools is Kovacic Algorithm, see \cite{ko}. Theoretical aspects of differential Galois theory can be found in the references \cite{IK,EK,MV}, while some applications applications of Differential Galois Theory in classical and quantum mechanics can be found  in the references \cite{Prim,acbl,AMW,Jmorales}. The algebraic geometry is the heart of differential Galois in the modern language, see \cite{EK}, also can be used to compute solutions of differential equations with parameters though Gr\"obner basis, see \cite{cox}.\\

\noindent In this paper we present a new computational method to obtain the explicit values of energy and wave functions for Schr\"odinger equations with polynomial potentials of degree even. The method is based in the characterization of differential Galois groups for Schr\"odinger equations with polynomial potentials, see \cite{acbl,AMW}, combined by the first time with the using of Gr\"obner basis. This paper is improves some results presented in \cite{hv}, which was  inspired in \cite{Prim,acbl,AMW}.\medskip

\noindent The structure of this paper is as follows. Section 1 contains the basic preliminaries to understand the rest of the paper and it is divided in three subsections. The first subsection is concerning the preliminaries of differential Galois theory, also know as Picard-Vessiot theory, following \cite{IK,MV}. The second subsection is concerning the theoretical an computational aspects related with Gr\"obner basis following \cite{cox1}. The  third subsection is devoted to non-relativistic quantum mechanics from a galoisian point of view following \cite{Prim,AMW}. Section 2 contains the results of this paper, which includes the study of Liouvillian solutions of Schr\"odinger equations with polynomial potentials of degree 4, 6, 8, 10, 12 and 14. Section 3 corresponds to final remarks and conclusions of the paper.

\section{Theoretical Background}
This section is devoted to the basic concepts necessaries to understand the rest of the paper. Preliminaries in this section includes theoretical aspects of Differential Galois Theory, Gr\"obner basis theory and galoisian approach to the integrability of Schr\"odinger equation.

\subsection{Differential Galois Theory} 

\ali A \textit{derivation} on a field $K$ is a map  $':K\rightarrow K$ satisfying the following properties, for all $a,b \in K$:
\begin{enumerate}
\item $(a+b)'=a'+b'$ and
\item $(ab)'=ab'+a'b$.
\end{enumerate}
A field $K$ equipped whit a derivation is named a \emph{differential field}. A field $F\supset K$ is a \emph{differential extension} of the differential field $K$ if the derivation of $K$ extends to a derivation in $F$.

\noindent An element $c\in K$, being  $K$ be a differential field, is called a \emph{constant} if $c'=0$. Moreover, the set of all constants of a field $K$ is a subfield of $K$, this field is called the \emph{field of constants} and will denoted by $C_{K}$.

\noindent On the other hand, there exist an analogue concept of the splitting field of a polynomial on the differential Galois theory, and it is called the Picard-Vessiot extension of a homogeneous linear differential equation (HLDE). Given an HLDE $\mathcal{L}(y)=0$ of order $n$ over a differential field $K$, a differential extension $F$ is a Picard-Vessiot extension if:
\begin{enumerate}
\item $C_{F}=C_{K}$,
\item The $C_{F}$-vector space $V$ of solutions of $\mathcal{L}(y)=0$ has dimension $n$,
\item $F$ is generated over $K$ by the solutions of $\mathcal{L}(y)=0$.
\end{enumerate}
Finally, we are able to describe a Liouvillian function. Let $(K,')$ be a differential field. An extension $L/K$ is said \emph{Liouvillian} if $C_{K}=C_{L}$ and there exist a tower of differential fields $K=K_{0}\subset K_{1}\subset\cdots\subset K_{n}=L$ such that $K_{i}=K_{i-1}(t_{i})$, $i\in\{0,1,\dots,n\}$ where
\begin{enumerate}
\item $t_{i}\in K_{i-1}$, \textit{i.e.}, $t_{i}$ is an integral, or
\item $t_{i}\neq 0$ and $\frac{t_{i}'}{t_{i}}\in K_{i-1}$, \textit{i.e.}, $t_{i}$ is an expotential, or
\item $t_{i}$ is an algebraic element over $K_{i-1}$.
\end{enumerate}
%\begin{definition}
Let $F/K$ be the Picard-Vessiot extension associated to the HLDE $\mathcal{L}(y)=0$. We say that the solutions of $\mathcal{L}(y)=0$ are \emph{Liouvillian} if there exist a Liouvillian extension $L/K$ such that $K\subset F\subset L$. Another important structure in this paper is the differential Galois group of $F/K$, it is the set of all differential $K$-automorphisms, that is $Gal(F/K)=\{\sigma:F\rightarrow F/ \sigma|_K=e\}$. The following theorem it is very well known in Picard-Vessiot theory, see \cite{IK,MV}.\medskip

\begin{theorem}
The component identity component of the differential Galois group associated to the HLDE $\mathcal{L}(y)=0$ is solvable if and only if the HLDE $\mathcal{L}(y)=0$ has Liouvillian solutions.
\end{theorem}

\noindent Along this paper, we consider the differential field $K=\mathbb{C}(x)$ and in consequence the field of constants will be $C_K=\mathbb{C}$.
%\end{definition}

\subsection{Gr\"obner Basis}
In this subsection we describe briefly what are Gr\"obner basis, following \cite{cox1}, which are very useful for our main study. We start studying a several variables polynomial system $f_{i}=0$. We can consider the \emph{ideal} generated by this set of polynomial and the idea is to find the best set of polynomial generating the same ideal (which means that they will have the same set of solutions, up to multiplicity) but in a triangular form. Formal definitions and properties of Gr\"obner basis can be found in \cite{cox1}. For our purposes, we summarize such definitions and properties as follows. 

%\begin{definition}
\noindent Let $I\subseteq k[x_{1},\dots,x_{n}]$ be a nontrivial ideal, and fix a monomial ordering on the ring of polynomials $k[x_{1},\dots,x_{n}]$. Then we denote by $LT(I)$ the set of leading terms of nonzero elements of $I$ and $\langle LT(I)\rangle$ the ideal generated by the elements in $LT(I)$.
%\end{definition}
\ali The following theorem will guaranty the existence of Gr\"obner basis, which is also useful to prove the \emph{Hilbert Basis Theorem} which states that every ideal $I\subseteq k[x_{1},\dots,x_{n}]$ has a finite generating set (the so called \emph{ideal description problem}).\medskip

\begin{theorem}\label{gb1}
Let $I\subseteq k[x_{1},\dots,x_{n}]$, $I\neq \{0\}$, be an ideal. Then the following statements hold.
\begin{enumerate}
\item $\langle LT(I)\rangle$ is a monomial ideal.
\item There exist $g_{1},\dots,g_{s}\in I$ such that $\langle LT(I)\rangle=\langle LT(g_{1}),\dots,LT(g_{s}) \rangle$.
\end{enumerate}
\end{theorem}
\medskip

\ali We recall that not all basis, such as the basis described in Theorem \ref{gb1}, have the same behavior.  
%\begin{definition}
Fix a monomial order on the polynomial ring $k[x_{1},\dots,x_{n}]$. A finite subset $G\{g_{1},\dots,g_{s}\}$ of an ideal $I\subseteq k[x_{1},\dots,x_{n}]$ different from zero ideal is said to be a \emph{Gr\"obner basis} if  $\langle LT(I)\rangle=\langle LT(g_{1}),\dots,LT(g_{s}) \rangle$.
%\end{definition}
\ali Now, the most important fact (for our purpose) is that a Gr\"obner basis is indeed a basis.
\begin{proposition}\label{pgb1}
Fix a monomial order. Then every ideal $I\subseteq k[x_{1},\dots,x_{n}]$ has a Gr\"obner basis. Furthermore, any Gr\"obner basis for an ideal $I$ is a basis of $I$.
\end{proposition}
\ali It could happen that $G=\{1\}$, this means that the ideal $I$ is indeed the polynomial ring  $k[x_{1},\dots,x_{n}]$ and obviously we won't have any solution for our system. This will be very useful to determine nonintegrability in the next section. More information about Gr\"obner basis can be found in \cite{cox1}.

\subsection{Schr\"odinger Equation}
In this paper we are interested in the stationary one-dimensional non-relativistic  Schr\"odinger equation 
\begin{equation}\label{refeq33}
H\psi (z)=\lambda\psi(z).
\end{equation}
Where $H=-\partial+V(z)$ is the so called Schr\"odinger operator, $z$ is a coordinate on the real line, the eigenfunction $\psi$ is the wave function, the eigenvalue $\lambda$ is the energy and $V(z)$ is the energy potential or potential. Along this paper we will be interested on polynomial potentials, it means that $V(z)$ will be an element on the ring $\c[z]$.

\ali A funtion $\psi$ will be on the ${\rm L}_{2}$ space when its norm,
\begin{equation}
||\psi(z)||=\int_{-\infty}^{+\infty}|\psi(z)|^{2}dz
\end{equation}
is well defined. The point spectrum of $H$ consist of all the values of $\lambda$ for which exist a non-null solution in ${\rm L}_{2}$ to the Schr\"odinger equation with the operator $H$. A solution $\psi_{n}(z)$ is a bound state if $\lambda_{n}$ belong to the point spectrum of $H$ and its norm is finite, it means,
\begin{equation}
\lambda_{n}\in Spec_{p}(H)\hspace{5pt}and\hspace{5pt}\int_{-\infty}^{+\infty}|\psi_{n}(z)|^{2}dz<\infty,n\in\z_{+}
\end{equation}
Some times the point spectrum sets as a crescent sentence $\lambda_{0}<\lambda_{1}\cdots<\lambda_{n}<\cdots$ of eigenvalues, in that case let $\psi_{0}<\psi_{1}<\cdots<\psi_{n}<\cdots$ be the bound state wave funtions with energy $\lambda_{0}<\lambda_{1}\cdots<\lambda_{n}<\cdots$. The wave function $\psi_{0}$ which corresponds the minimun level of energy is called ground state while the others one are called excited states.

\begin{definition}
We named the set of all eigenvalues which the equation  \eqref{refeq33} is integrable in the Liouville sense, algebraic spectrum and we denote by $\Lambda\subseteq\mathbb{C}$

\end{definition}
\begin{definition}
We say that a potential $V(x)\in\mathbb{C}[x]$ is:

\begin{enumerate}
	\item Algebraically solvable if $\Lambda$ is infinite, or
	\item Algebraically quasi-solvable if $\Lambda$ is finite, or
	\item Algebraically non-solvable if $\Lambda=\emptyset$
	
\end{enumerate}
\end{definition}
\begin{example}

\begin{enumerate}
	\item $V(x)=0$, then $\Lambda=\mathbb{C}$, consequently $V(x)$ is a potential algebraically solvable.
	\item $V(x)=x^{6}+3x^{2}$, then $\Lambda=\{0\}$, consequently $V(x)$ is a potential algebraically quasi-solvable.
	\item $V(x)=x$, then $\Lambda=\emptyset$, consequently $V(x)$ is a potential algebraically non-solvable.
	\end{enumerate}
\end{example}

\noindent Theorem \ref{caracterization}, know as the Galoisian characterization of polynomial potentials theorem, is the main tool in this work and it allows us make a characterization of quasi-solvable polynomial potential, this theorem requires a suitable form of the potential $V(z)$ so we shall enunciate a lemma as a complent to solve this problem 
\begin{lemma}\cite[lemma 2.4, p.275]{acbl}
Every even degree monic polynomial can be written in one only way completing squares, that is,
\begin{equation}
P_{2n}(z)=z^{2n}+\sum_{i=0}^{2n-1}a_{i}z^{i}=\Big(z^{n}+\sum_{i=0}^{n-1}b_{i}z^{i} \Big)^{2}+\sum_{i=0}^{n-1}c_{i}z^{i}.
\end{equation}
\end{lemma} 
\begin{theorem}\cite[theorem 2.5, p.276]{acbl}\label{caracterization} Let us consider the
Schr\"odinger equation \eqref{refeq33}, with $V(x)\in\mathbb{C}[x]$ a polynomial of degree $k>0$. Then, its differential Galois group  $DGal(L_{\lambda}/K)$ falls in one of the following cases:

\begin{enumerate}
	\item  $DGal(L_{\lambda}/K) =SL(2,\mathbb{C})$,
	\item   $DGal(L_{\lambda}/K) =\mathbb{B}$.
\end{enumerate}

\ali Furthermore,  $DGal(L_{\lambda}/K) =\mathbb{B}$ if and only if the following conditions hold:

\begin{enumerate}
	\item $V(x)-\lambda$ is a polynomial of degree $k=2n$ writing in its completing square form
	\item $c_{n-1}-n$ o $-c_{n-1}-n$ is a positive even number $2s$, $s\in\mathbb{Z}_{+}$.
	\item There exists a monic polynomial $P_{s}$  of degree $s$, satisfying:
	
	\begin{eqnarray*}
		 &\partial_{x}^{2}P_{s}+2(x^{n}+\sum_{k=0}^{n-1}{b_{k}x^{k}})\partial_{x}P_{s}+(nx^{n-1}+\sum_{k=0}^{n-2}{(k+1)b_{k+1}x^{k}}-\\ & \sum_{k=0}^{n-1}{c_{k}x^{k}})P_{s}=0,& or\\
		 &\partial_{x}^{2}P_{s}-2(x^{n}+\sum_{k=0}^{n-1}{b_{k}x^{k}})\partial_{x}P_{s}-(nx^{n-1}+\sum_{k=0}^{n-2}{(k+1)b_{k+1}x^{k}}+\\ &\sum_{k=0}^{n-1}{c_{k}x^{k}})P_{s}=0.&
	\end{eqnarray*}
	\end{enumerate}
	
\ali	In such cases, the only possibilities for eigenfunctions with polynomial potentials
are given by

	\begin{eqnarray*}
	\psi_{\lambda}=P_{s}e^{f(x)}, & or & \Psi_{\lambda}=P_{s}e^{-f(x)},\,\,  where \,\, f(x)= \frac{x^{n+1}}{n+1}+\sum_{k=0}^{n-1}\frac{b_{k}x^{k+1}}{k+1}.
	\end{eqnarray*}
\end{theorem}

\begin{corollary}
Assume that $V(x)$ is an algebraically solvable polynomial potential. Then, $V(x)$ is a polynomial of degree $2$.
\end{corollary}

\begin{corollary}
Suppose $V(x)$ a polynomial potential of odd degree. Then \eqref{refeq33} is not integrable.
\end{corollary}

\section{Quasi-Solvable Polynomial Potentials}
Our goal in this section is to show how useful  Gr\"obner basis are to find the solution of the problem of Liouvillian integrability for polynomial potentials, we present a simple algorithm that together to theorem \ref{caracterization} reduce our initial problem to  algebraic geometry. The following proposition characterizes the Gr\"obner basis associated to Schr\"odinger equations with polynomial potentials.\\

\begin{proposition}\label{propgrob}
Let $I$ be the ideal generated by the set of polynomial conditions obtained in item 3 of theorem \ref{caracterization} and let $G$ be a Gr\"obner basis of $I$. Then, the points of the affine variety $V(G)$ are in correspondence with the coefficients of quasi-solvable potentials.
\end{proposition}
\begin{proof}
The set of polynomial generating $I$ are dependent of the coefficients of $V(z)=z^{2n}+\sum_{k=0}^{2n-1}a_{k}z^{k}$, so  $I=\{f_{i}(a_{0},\dots,a_{2n-1})\}$. By proposition \ref{pgb1} the affine variety $V(G)=V(I)$.
\end{proof}

\ali As application of Theorem \ref{caracterization} and Proposition \ref{propgrob} we present the following algorithm that can be implemented easier into a computational algebra software.\\

\begin{algorithm}\label{algrob}
\begin{enumerate}
\item Set the potential and write it in completing the square form.
\item Set a value for $s$.
\item Formulate the equations described on step three on theorem \ref{caracterization}.
\item Set a polynomial with unknown coefficients.
\item Replace polynomial described in last step in one of the equations formulated in step 3.
\item Set an homogeneous polynomial system with the coefficients of the polynomial obtained in last step.
\item Compute the Gr\"obner basis of the ideal generated by the system obtained in the previous step.
\end{enumerate}
\end{algorithm}

\begin{remark}
\ali Steps 1 to 3 are a procedure related to the application of the theorem \ref{caracterization} while steps 4 to 6 establishes polynomial conditions dependent of the coefficients of the potential in order to grant the existence of polynomial solutions $P_{s}$ to the equations described on the item 3 of the theorem \ref{caracterization}. Step 7 is the most important and corresponds to Proposition \ref{propgrob}, it calculates the simplest way of express these conditions via the Gr\"obner basis $G$. Thus, the quasi-solvable potentials will be exactly those which coefficients are points on the algebraic variety $V(G)$ generated by $G$. To illustrate this algorithm we present the figure 1.
\end{remark}

\ali Now we apply Proposition \ref{propgrob} and for instance Algorithm \ref{algrob} to compute the spectral polynomials concerning Quartic, Sextic, Octic, Decatic, Dodecatic and Tetrakaidecatic potentials.

\subsection{Quartic Potential}
Let us consider the potential $V(x)=x^{4}+4x^{3}+2x^{2}-\mu x$, so applying completing the square (see figure \ref{degree4}) we can writte this potential as follows,
\begin{equation}\label{quarticp}
V(x)=(x^{2}+2x-1)^{2}+(4-\mu)x-1.
\end{equation}
Now, by virtue of theorem \ref{caracterization} we have $\pm(4-\mu)-2=2s$, where $s\in\z_{+}$.  Consequantly, $\mu$ is a discrete parameter that can be $2-2s$ or either $6+2s$. As a third step in Galoisian characterization theorem, in order to determinate Liouvillian integrability for the Schr\"odinger equation associated to potential \eqref{quarticp}, it must exist a monic polynomial $P_{s}$ of degree $s$ satisfying at least one of the following equations
\begin{align}
P_{s}''+2(x^{2}+2x-1)P_{s}'+((\mu-2)x+3+\lambda)P_{s}&=0,\\
P_{s}''-2(x^{2}+2x-1)P_{s}'+((\mu-6)x-1+\lambda)P_{s}&=0.
\end{align}
Now, by algorithm \ref{algrob} we obtain the results summarized in Table 1  and in Table 2.
\begin{table}[h]\label{table1}
\caption{The spectral polynomial of potential \eqref{quarticp} for the case $\mu=2-2s$.}
\begin{tabular}{|l|l|l|}
\hline
 $s$&$\mu=2-2s$ & $T(s,\lambda)$ \\ \hline\hline
 $0$& $2$& $\lambda-1$ \\ \hline
 $1$& $0$& $\lambda^{2}+10\lambda+17$ \\ \hline
 $2$&$-2$& $\lambda^{3}+21\lambda^{2}+115\lambda+135$ \\ \hline
 $3$&$-4$& $\lambda^{4}+36\lambda^{3}+406\lambda^{2}+1572\lambda+1521$ \\ \hline
 $4$&$-6$& $\lambda^{5}+55\lambda^{4}+1050\lambda^{3}+8366\lambda^{2}+26613\lambda+27659$ \\ \hline
 $5$&$-8$& $\lambda^{6}+78\lambda^{5}+2255\lambda^{4}+30276\lambda^{3}+196015\lambda^{2}+596046\lambda+777825$ \\ \hline
\end{tabular}
\end{table}

\begin{table}[h]\label{table2}
\caption{The spectral polynomial of potential \eqref{quarticp} for the case $\mu=6+2s$.}
\begin{tabular}{|l|l|l|}
\hline
 $s$& $\mu=6+2s$ & $T(s,\lambda)$ \\ \hline\hline
$0$ &$6$ &$\lambda-1$  \\ \hline
$1$ &$8$ &$\lambda^{2}-6 \lambda+1$  \\ \hline
$2$ &$10$&$\lambda^{3}-15 \lambda^{2}+43 \lambda+51$ \\ \hline
$3$ &$12$&$\lambda^{4}-28 \lambda^{3}+214 \lambda^{2}-156 \lambda-1615$  \\ \hline
$4$ &$14$&$\lambda^{5}-45 \lambda^{4}+650 \lambda^{3}-2634 \lambda^{2}-8027 \lambda+41799$  \\ \hline
$5$ &$16$&$\lambda^{6}-66 \lambda^{5}+1535 \lambda^{4}-13404 \lambda^{3}+3343 \lambda^{2}+428670 \lambda-984879$  \\ \hline
\end{tabular}
\end{table}
\ali Some values in $\Lambda$ are easily computable, for example: for $\mu=2-2s$ we have the results summarized in Table 3,
\begin{table}[H]\label{table3}
\caption{Some values in $\Lambda$ for $\mu=2-2s$  in potential \eqref{quarticp}.}
\begin{tabular}{|l|l|l|}
\hline
$s$ & $\Lambda$  \\ \hline\hline
0 & $\{1\}$   \\ \hline
1 & $\{-2 \sqrt{2}-5, 2 \sqrt{2}-5\}$   \\ \hline
2 & \footnotesize$\begin{multlined}\{\frac{2 \sqrt[3]{-9+i \sqrt{1455}}}{3^{2/3}}+\frac{16}{\sqrt[3]{3 \left(-9+i \sqrt{1455}\right)}}-7,-\frac{\left(1+i \sqrt{3}\right) \sqrt[3]{-9+i \sqrt{1455}}}{3^{2/3}}-\frac{8 \left(1-i \sqrt{3}\right)}{\sqrt[3]{3 \left(-9+i \sqrt{1455}\right)}}-7,\\ -\frac{\left(1-i \sqrt{3}\right) \sqrt[3]{-9+i \sqrt{1455}}}{3^{2/3}}-\frac{8 \left(1+i \sqrt{3}\right)}{\sqrt[3]{3 \left(-9+i \sqrt{1455}\right)}}-7\}\end{multlined}$   \\ \hline
\end{tabular}
\end{table}
\ali On the other hand, for $\mu=6+2s$ we have the results summarized in Table 4., 
\begin{table}[H]
\caption{Some values in $\Lambda$ for $\mu=6+2s$  in potential \eqref{quarticp}.}
\begin{tabular}{|l|l|l|}
\hline
$s$& $\Lambda$   \\ \hline
$0$ & $\{1\}$   \\ \hline
$1$ & $\{3-2 \sqrt{2},2 \sqrt{2}+3 \}$   \\ \hline
$2$ & \footnotesize$\begin{multlined}\{\frac{2 \sqrt[3]{-9+i \sqrt{1455}}}{3^{2/3}}+\frac{16}{\sqrt[3]{3 \left(-9+i \sqrt{1455}\right)}}+5,-\frac{\left(1+i \sqrt{3}\right) \sqrt[3]{-9+i \sqrt{1455}}}{3^{2/3}}-\frac{8 \left(1-i \sqrt{3}\right)}{\sqrt[3]{3 \left(-9+i \sqrt{1455}\right)}}+5,\\ \}-\frac{\left(1-i \sqrt{3}\right) \sqrt[3]{-9+i \sqrt{1455}}}{3^{2/3}}-\frac{8 \left(1+i \sqrt{3}\right)}{\sqrt[3]{3 \left(-9+i \sqrt{1455}\right)}}+5\end{multlined}$   \\ \hline
\end{tabular}
\end{table}
\ali in addition, the solutions of the Schr\"odinger equation is given by 
\begin{equation}
\psi_{\lambda}(x)=\begin{cases} P_{s}e^{\frac{x^3}{3}+x^2-x} & \text{ if } \mu=2-2s,\\
P_{s}e^{-\frac{x^3}{3}-x^2+x} & \text{ if }  \mu=6+2s.\end{cases}
\end{equation}
\subsection{Sextic Potential}
Let us consider the nonsingular turbiner potential $x^{6}-(4J\mp 1)x^{2}$, where $J$ is a nonnegative integer, this potential has been studied by Bender and Dunne in \cite{BeDu}, they showed that there exist a correspondence between the solutions of the Schr\"odinger equation associated to this potential and sets of orthogonal polynomials ${P_{s}}$. In general we can consider non-monic polynomials of degree six and after we use the completing square  method (see figure \ref{degree6}). In this paper we shall distinguish two cases for the sextic anharmonic potentials:\\

\textbf{Case 1:} Let us set
\begin{equation}\label{turbiner}
\psi''(x)=(x^{6}-(4J- 1)x^{2}-\lambda)\psi(x),
\end{equation}
as our main object. Due to step two in theorem \ref{caracterization}, we can conclude $\mp(4J-1)-3=2s$ where $s\in\mathbb{Z}_{+}$. If $-(4J-1)-3=2s$, then $J=-\frac{s+1}{2}$ which is not possible because $J$ is a nonnegative integer. So $(4J-1)-3=2s$, \textit{i.e.}, $J=\frac{s+2}{2}$, therefore  $s$ takes nonnegative even values.

\ali Then, in order to achieve Liouvillian integrability, there must exist a monic polynomial $P_{s}(x)$ satisfying the equation
\begin{equation}\label{odeauxsex1}
	P_{s}''-2x^{3}P_{s}'-(3x^{2}-\lambda-(4J-1)x^{2})P_{s}=0.
\end{equation}
Algorithm \ref{algrob}, give us a tool to find spectral polynomials to equation \eqref{odeauxsex1}, which is summarized in Table 5.

\begin{table}[h]\label{table5}
\caption{Spectral polynomials of equation \eqref{odeauxsex1}.}
\centering
\begin{tabular}{|l|l|l|}
\hline
$s$ & $J$ & $T(s,\lambda)$ \\ \hline\hline
 0& 1&  $\lambda$\\ \hline
 2& 2&$\lambda^{2}-8$  \\ \hline
 4& 3&$\lambda^{3}-64 \lambda$  \\ \hline
 6& 4&$\lambda^{4}-240 \lambda^{2}+880$  \\ \hline
 8& 5&$\lambda^{5}-640 \lambda^{3}+47104 \lambda$  \\ \hline
 10&6&$\lambda^{6}-1400 \lambda^{4}+331456 \lambda^{2}+5184000$  \\ \hline
\end{tabular}
\end{table}
\ali On the other hand, the solutions of equations \eqref{turbiner} and \eqref{odeauxsex1} are easily calculable once the zeroes of $T(s,\lambda)$ are known, for example:
The Gr\"obner basis generated by replace $P_{6}=a x^5+b x^4+c x^3+d x^2+e x+g+x^6$ in \eqref{odeauxsex1} is,
\begin{equation*}
\{2880 - 240 \lambda^2 + \lambda^4, 384 g - 216 \lambda + \lambda^3, e, 
 120 + 32 d - \lambda^2, c, 4 b + \lambda, a\}
\end{equation*}
we conclude that the coefficients $g$, $d$, and $b$ are polynomials in $\lambda$ and that any other coefficients is zero. In addition, the solution to the Schr\"odinger equation is given by 
\begin{equation}
\psi_{\lambda}(x)=P_{s,\lambda}(x)e^{\frac{x^{4}}{4}},
\end{equation}
which is summarized in Table 6.
\begin{table}[h]
\caption{Solutions to \eqref{turbiner} and \eqref{odeauxsex1} for small values of $s$. }
\centering
\begin{tabular}{|l|l|l|l|}
\hline
$s$ & $\Lambda$ & $P_{s,\lambda}$ &$\psi_{\lambda}(x)$\\ \hline\hline
0 & $\{0\}$ & 1 &$exp(\frac{x^{4}}{4})$\\ \hline
2 & $\{\mp2\sqrt{2}\}$ & $x^{2}\pm\frac{1}{\sqrt{2}}$ &$(x^{2}\pm\frac{1}{\sqrt{2}})exp(\frac{x^{4}}{4})$\\ \hline
4 &$\{0,\mp 8\}$  & $x^{4}-\frac{3}{2}$, $x^{4}\pm 2x^{2}+\frac{1}{2}$ &$(x^{4}-\frac{3}{2})exp(\frac{x^{4}}{4}),(x^{4}\pm 2x^{2}+\frac{1}{2})exp(\frac{x^{4}}{4})$\\ \hline
\end{tabular}
\end{table}

\paragraph{Case 2:} In this case, let us set
\begin{equation}\label{turbiner2}
	\psi''(x)=(x^{6}-(4J+ 1)x^{2}-\lambda)\psi(x),
\end{equation}
as our object of study. In a similar way to the above case we can conclude that $J=\frac{s+1}{2}$, therefore $s$ takes nonnegative odd values. And by theorem \ref{caracterization}, there must exist a monic polynomial $P_{s}(x)$ satisfying the  equation,
\begin{equation}\label{odeauxsex2}
P_{s}''-2x^{3}P_{s}'-(3x^{2}-\lambda-(4J+1)x^{2})P_{s}=0.
\end{equation}
A simple application of algorithm \ref{algrob}, give us a list of the spectral polynomial for each value of $s$ (see Table 7).
\begin{table}[H]
\caption{Spectral polynomials of equation \eqref{odeauxsex2}.}
\centering
\begin{tabular}{|l|l|l|}
\hline
$s$ & $J$ & $T(s,\lambda)$ \\ \hline\hline
 1& 1&  $\lambda$\\ \hline
 3& 2&$\lambda^{2}-24$  \\ \hline
 5& 3&$\lambda^{3}-128 \lambda$  \\ \hline
 7& 4&$\lambda^{4}-400 \lambda^{2}+12096$  \\ \hline
 9& 5&$\lambda^{5}-960 \lambda^{3}+129024 \lambda$  \\ \hline
\end{tabular}
\end{table}
\ali In this case, the solution to the  Schr\"odinger equation \eqref{turbiner2} is given by
\begin{equation}
\psi_{\lambda}(x)=P_{s,\lambda}e^{-\frac{x^{4}}{4}},
\end{equation}
which is summarized in Table 8.
\begin{table}[H]
\caption{Solutions to equation \eqref{odeauxsex2} for small values of $s$. }
\centering
\begin{tabular}{|l|l|l|}
\hline
$s$ & $\Lambda$ &  $P_{s,\lambda}$\\ \hline\hline
 1&$\{0\}$  & $x$ \\ \hline
 3&$\{\pm 2\sqrt{6}\}$  & $x^{3}\mp\frac{\sqrt{6}}{2}x$  \\ \hline
 5&$\{0,\pm 8\sqrt{2}\}$  & $x^{5}-\frac{5}{2}x, x^{5}\mp 2\sqrt{2}x^{3}+\frac{3}{2}x$  \\ \hline
\end{tabular}
\end{table}

\subsection{Octic Potential}
Let us consider the potential $V(x)=x^8+(2 \delta +4) x^4+\mu  x^3+\delta ^2+4 \delta +4$, this potential can be written in the following form via completing the square (see figure \ref{degree8}):
\begin{equation}\label{octicp}
V(x)=(x^4+\delta +2)^{2}+\mu  x^3
\end{equation}
Now, we can use theorem \ref{caracterization} to determine conditions over the parameters $\delta$ and $\mu$ in order to achieve Liouvillian integrability. First, one can say that $\mu$ is a discrete parameter that can be $2s+4$ or either $-2s-4$ where $s$ is an nonnegative integer. Secondly, there must exist a monic polynomial $P_{s}$ of degree $s$ satisfying at least one of the following equations:
\begin{eqnarray}
 P_{s}''+2(x^4+\delta +2) P_{s}'+(-2 s x^{3}+\lambda)P_{s}&=0\label{odeauxoct1}\\
  P_{s}''-2(x^4+\delta +2) P_{s}'-(-2 s x^{3}-\lambda)P_{s}&=0\label{odeauxoct2}
\end{eqnarray}
Algorithm \ref{algrob} provide us a tool to calculate (if they exist) the polynomial solutions of above equations (see Table 9 and Table 10).
\begin{table}[h]
\caption{Suitable parameters and solution for equation \eqref{odeauxoct1}.}
\centering
\begin{tabular}{|l|l|l|l|l|}
\hline
s & $\mu$ & $\delta$ & $\Lambda$ & $P_{s}$  \\ \hline\hline
0 & 4 & $\c$ & 0 & 1 \\ \hline
1 & 6 & -2 & 0 & $x$ \\ \hline
2 &\multicolumn{4}{l|}{Not integrable}  \\ \hline
3 &\multicolumn{4}{l|}{Not integrable}  \\ \hline
4 &\multicolumn{4}{l|}{Not integrable}  \\ \hline
5 & 14 & -2 & 0 & $x^{5}+2$ \\ \hline
\end{tabular}
\end{table}
\begin{table}[h]
\caption{Suitable parameters and solution for equation \eqref{odeauxoct2}.}
\centering
\begin{tabular}{|l|l|l|l|l|}
\hline
s & $\mu$ & $\delta$ & $\Lambda$ & $P_{s}$  \\ \hline\hline
0 & -4 & $\c$ & 0 & 1 \\ \hline
1 & -6 & -2 & 0 & $x$ \\ \hline
2 &\multicolumn{4}{l|}{Not integrable}  \\ \hline
3 &\multicolumn{4}{l|}{Not integrable}  \\ \hline
4 &\multicolumn{4}{l|}{Not integrable}  \\ \hline
5 & -14 & -2 & 0 & $x^{5}-2$ \\ \hline
\end{tabular}
\end{table}

\ali In addition, the solutions to the Schr\"odinger equation associated to the potential \eqref{octicp} are given by:
\begin{equation}
\psi_{\lambda}(x)=\begin{cases} P_{s}e^{\frac{x^4}{4}+(\delta+2)x} & \text{ if } \mu=2s+4,\\
P_{s}e^{-\frac{x^4}{4}-(\delta+2)x} & \text{ if }  \mu=-2s-4.\end{cases}
\end{equation}

\ali We can conclude that algorithm \ref{algrob} is also a useful tool to determine non-integrability of certain systems.

\subsection{Decatic Potential}
In this section let us consider the potential $V(x)=x^{10}-x^8+x^6+\delta  x^4+ \epsilon x^2 $, which include the specific potentials studied by Chaudhuri and Mondal in \cite{ChMo}, there is a  approach developed by D. Brandon and N. Saad in \cite{BrNa} using asymptotic iteration method but in this work we will apply the theorem \ref{caracterization} in order to make a Galoisian approach.

\ali Above decatic potential can be written in the following way via completing the square (see figure \ref{degree10}):
\begin{equation}\label{decaticp}
V(x)=\bigg(x^{5}-\frac{x^{3}}{2}+\frac{3 x}{8}\bigg)^{2}+\left(\delta +\frac{3}{8}\right) x^4+ \left(\epsilon -\frac{9}{64}\right)x^2,
\end{equation}
in virtue of theorem \ref{caracterization}, in order to determine Liouvillian integrability, there must exist a monic polynomial $P_{s}$ of degree $s$, satisfying one of the following equations
\begin{eqnarray}
P_{s}''+2\left(x^5-\frac{x^3}{2}+\frac{3 x}{8}\right)P_{s}'+\bigg(\left(\frac{37}{8}-\delta \right) x^4+ \left(-\epsilon -\frac{87}{64}\right)x^2+\frac{3}{8} +\lambda\bigg)P_{s}&=0,\label{odeauxdec1}\\
P_{s}''-2\left(x^5-\frac{x^3}{2}+\frac{3 x}{8}\right)P_{s}-\bigg(\left(\delta +\frac{43}{8}\right) x^4+ \left(\epsilon -\frac{105}{64}\right)x^2+\frac{3}{8}-\lambda \bigg)P_{s}&=0.\label{odeauxdec2}
\end{eqnarray}
\begin{remark}
It is also clear from theorem \ref{caracterization} that $\delta$ is a number of the form $2s+\frac{37}{8}$ or $-2s-\frac{43}{8}$ where $s\in\z_{+}$.
\end{remark}
\ali An application of algorithm \ref{algrob}, give us values of parameters of the potential \eqref{decaticp}, see Table 11, Table 12 and Table 13.
\begin{table}[h]
\caption{Suitable parameters for equation \eqref{odeauxdec1}. }
\centering
\begin{tabular}{|l|l|l|l|}
\hline
$s$ & $\delta$ & $M_{s}(\epsilon)$ &  $\lambda$\\ \hline\hline
 0& $\frac{37}{8}$ & $-\epsilon -\frac{87}{64}$ &$-\frac{3}{8}$\\ \hline
 1& $\frac{53}{8}$ & $64 \epsilon +151$ &$-\frac{9}{8}$\\ \hline
 2& $\frac{69}{8}$ & $262144 \epsilon ^3+2117632 \epsilon ^2+6925504 \epsilon +17694023$ &$\frac{-4096 \epsilon ^2-19328 \epsilon -49425}{16384}$\\ \hline
 3& $\frac{85}{8}$ & $262144 \epsilon ^3+2904064 \epsilon ^2+11947200 \epsilon +43776519$ &$\frac{-4096 \epsilon ^2-27520 \epsilon -85137}{16384}$\\ \hline
\end{tabular}
\end{table}

\begin{table}[H]
\caption{Suitable parameters for equation \eqref{odeauxdec2}. }
\centering
\begin{tabular}{|l|l|l|l|}
\hline
$s$ & $\delta$ & $M_{s}(\epsilon)$ &  $\lambda$\\ \hline\hline
0 & $-\frac{43}{8}$ & $-\epsilon+\frac{105}{64}$ &$\frac{3}{8}$\\ \hline
1 & $-\frac{59}{8}$ & $64 \epsilon -169$ & $\frac{9}{8}$\\ \hline
2 & $-\frac{75}{8}$ & $262144 \epsilon ^3-2338816 \epsilon ^2+8178880 \epsilon -3037945$  & $\frac{4096 \epsilon ^2-21632 \epsilon +55185}{16384}$ \\ \hline
3 & $-\frac{91}{8}$ & $262144 \epsilon ^3-3125248 \epsilon ^2+13642944 \epsilon +2959431$  & $\frac{4096 \epsilon ^2-29824 \epsilon +93201}{16384}$\\ \hline
\end{tabular}
\end{table}
\begin{remark}
The zeroes of polynomial $M_{s}(\epsilon)$ are the suitable values of parameter $\epsilon$ of potential \eqref{decaticp} in order to obtain Liouvillian integrability.
\end{remark}
\ali The solutions of the equation \eqref{odeauxdec1} and \eqref{odeauxdec2} are easily computable once the value of $\epsilon$ sets.

\begin{table}[H]
\caption{Polynomial solutions $P_{s}$ }
\centering
\begin{tabular}{|l|l|l|}
\hline
$s$ & Eq. \eqref{odeauxdec1} & Eq. \eqref{odeauxdec2}\\ \hline\hline
0 & 1 &1\\ \hline
1 & $x$ &$x$\\ \hline
2 & $x^{2}-\frac{64\epsilon+215}{256}$& $x^{2}+\frac{64\epsilon-233}{256}$\\ \hline
3 &$x^{3}-\frac{64\epsilon+279}{256}x$&$x^{3}+\frac{64\epsilon-297}{256}x$\\ \hline
\end{tabular}
\end{table}
\ali The solutions of the Schr\"odinger equation $\psi_{\lambda}$ generated by the equation \eqref{odeauxdec2} give us bound states of the form
\begin{equation}
\psi_{\lambda}(x)=P_{s,\lambda}e^{-\frac{x^6}{6}+\frac{x^4}{8}-\frac{3 x^2}{16}}.
\end{equation}
In the other hand, we never obtain bound states from equation \eqref{odeauxdec1} but the solution $\psi_{\lambda}$ is given by,
\begin{equation}
\psi_{\lambda}(x)=P_{s,\lambda}e^{\frac{x^6}{6}-\frac{x^4}{8}+\frac{3 x^2}{16}}.
\end{equation}
\subsection{Dodecatic Potential} Consider the potential $V(x)=x^{12}+\kappa  x^6+\mu  x^5$, this potentical can be written in the following form using completing the square (see figure \ref{degree12}):
\begin{equation}\label{dodecaticp}
V(x)=\bigg(x^6+\frac{\kappa }{2}\bigg)^{2}+\mu  x^5-\frac{\kappa ^2}{4}
\end{equation}
By step two in theorem \ref{caracterization} we have that $\mu$ is a discrete parameter that can be $2s+6$ or either $-2s-6$. Now, by third step in theorem \ref{caracterization}, there must exist a monic polynomial $P_{s}$ of degree $s$ that satisfies  at least one of the following equations:
\begin{eqnarray}
P_{s}''+(2x^{6}+\kappa)P_{s}'+\left((6-\mu ) x^5+\frac{\kappa ^2}{4}+\lambda\right)P_{s}&=0\label{odeaux121}\\
P_{s}''-(2x^{6}+\kappa)P_{s}'-\left((6+\mu ) x^5-\frac{\kappa ^2}{4}-\lambda\right)P_{s}&=0\label{odeaux122}
\end{eqnarray}
A direct application of algorithm \ref{algrob} give us suitable values of parameters of the potential \eqref{dodecaticp} in order to determine Liouvillian integrability or establish nonintegrability, see Table 14.
\begin{table}[H]
\caption{Suitable values of parameters for equations \eqref{odeaux121} and  \eqref{odeaux122}. }
\centering
\begin{tabular}{|l|l|l|l|l|l|}
\hline
$s$ & $\mu=2s+6$  &$\mu=-2s-6$  &$\kappa$& $\lambda$ & $P_{s}$ \\ \hline\hline
0   & 6      &  -6&$\c$    & $-\frac{\kappa^{2}}{4}$ &1\\ \hline
1   & 8      &  -8&0       &0           &$x$\\ \hline
2   & \multicolumn{5}{l|}{Not integrable}  \\ \hline
3   & \multicolumn{5}{l|}{Not integrable}  \\ \hline
4   & \multicolumn{5}{l|}{Not integrable}  \\ \hline
\end{tabular}
\end{table}
\ali additionally, the solution to the Schr\"odinger equation associated to the potential \eqref{dodecaticp} is given by:
\begin{equation}
\psi_{\lambda}(x)=\begin{cases} P_{s}e^{\frac{x^7}{7}+\frac{\kappa}{2}x} & \text{ if } \mu=2s+6,\\
P_{s}e^{-\frac{x^7}{7}-\frac{\kappa}{2}x} & \text{ if }  \mu=-2s-6.\end{cases}
\end{equation}
\subsection{Tetrakaidecatic Potential}
Consider the potential $V(x)=(x^7+\delta +2)^2+\mu  x^6+\kappa  x^2$, this potential can be written in the following form using completing the square (see figure \ref{degree14}):
\begin{equation}\label{tetrakaidecaticp}
V(x)=(x^7+\delta +2)^{2}+\mu  x^6+\kappa  x^2
\end{equation}
It is clear that $\mu$ is a discrete parameter of the form $2s+7$ or either $-2s-7$. Now, in order to determine Liouvillian integrability, there must exist a monic polynomial $P_{s}$ of degree $s$ which satisfies some of the following equations:
\begin{eqnarray}
P_{s}''+2(x^7 + \delta + 2)P_{s}'+(-2 s x^6-\kappa  x^2+\lambda)&=0\label{odeaux141}\\
P_{s}''-2(x^7 + \delta + 2)P_{s}'-(-2 s x^6+\kappa  x^2-\lambda)&=0\label{odeaux142}
\end{eqnarray}
Applying algorithm \ref{algrob} to this equations, we will obtain suitable parameters for potential \eqref{tetrakaidecaticp} or we will determine nonintegrability for the associated Schr\"odinger equation, see Table 15.
\begin{table}[H]
\caption{Suitable values of parameters for equations \eqref{odeaux141} and  \eqref{odeaux142}. }
\centering
\begin{tabular}{|l|l|l|l|l|l|l|}
\hline
$s$ & $\mu=2s+7$  &$\mu=-2s-7$ &$\delta$ &$\kappa$& $\lambda$ & $P_{s}$ \\ \hline\hline
0   & 7           &  -7        & $\c$    & 0      & 0         &1\\ \hline
1   & 9           &  -9        &  -2     & 0      &0          &$x$\\ \hline
2   & \multicolumn{6}{l|}{Not integrable}  \\ \hline
3   & \multicolumn{6}{l|}{Not integrable}  \\ \hline
4   & \multicolumn{6}{l|}{Not integrable}  \\ \hline
\end{tabular}
\end{table}
\ali In addition, the solution to the Schr\"odinger equation associated to the tetrakaidecatic potential \eqref{tetrakaidecaticp} is given by:
\begin{equation}
\psi_{\lambda}(x)=\begin{cases} P_{s}e^{\frac{x^8}{8}+(\delta+2)x} & \text{ if } \mu=2s+7,\\
P_{s}e^{-\frac{x^8}{8}-(\delta+2)x} & \text{ if }  \mu=-2s-7.\end{cases}
\end{equation}

\section{Appendix A}

\begin{figure}[H]
\centering
\caption{Example of algorithm \ref{algrob} for a decatic potential using Mathematica.}
\label{exampmath}
\includegraphics[width=10cm, height=15cm]{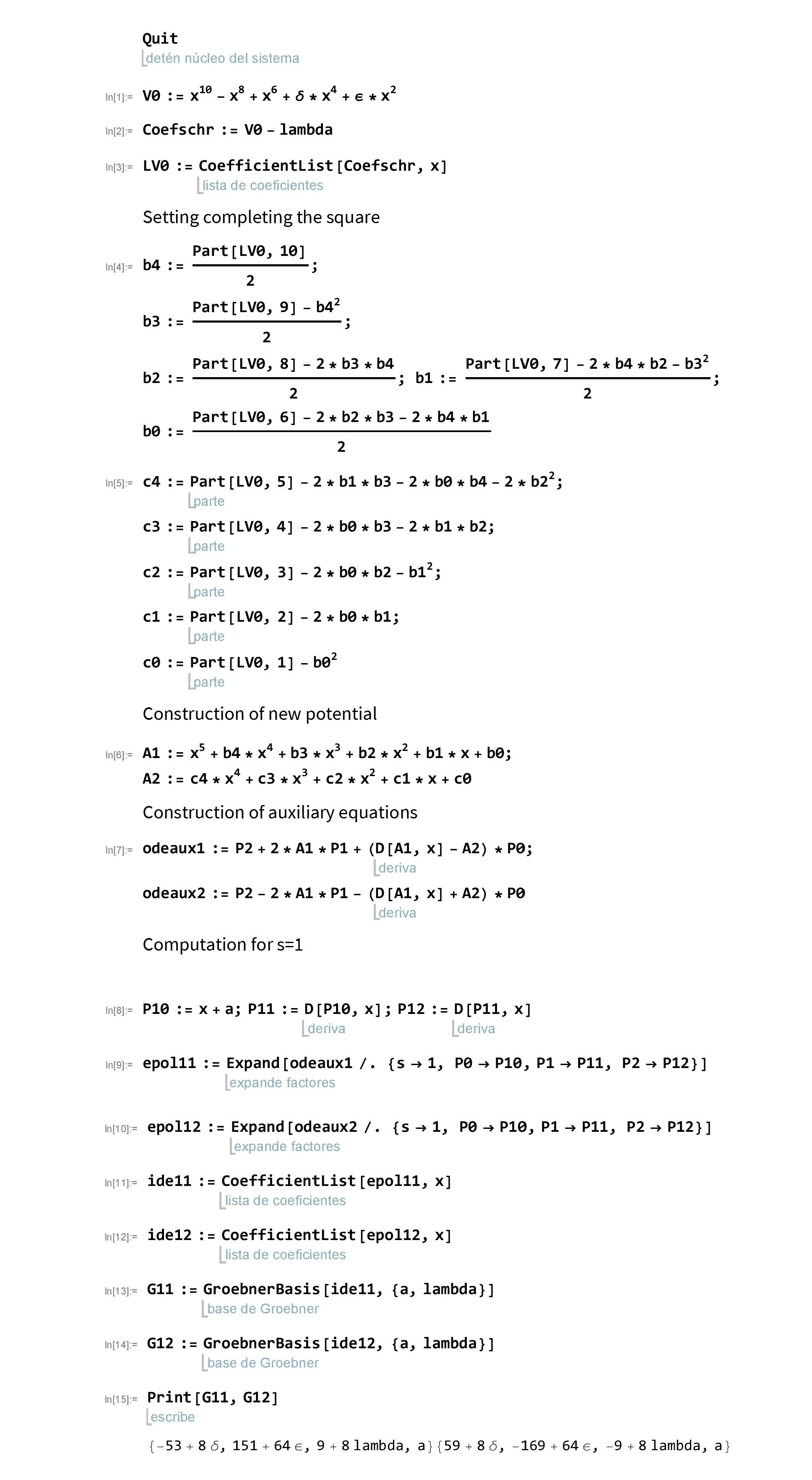}
\end{figure}

\begin{figure}[H]
\centering
\caption{Completing the square for degree 4}
\label{degree4}
\includegraphics[scale=0.4]{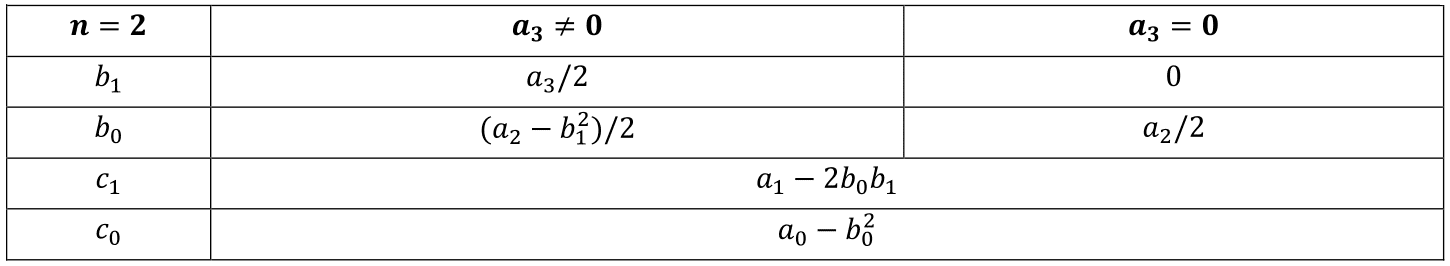}
\end{figure}

\begin{figure}[H]
\centering
\caption{Completing the square for degree 6}
\label{degree6}
\includegraphics[scale=0.4]{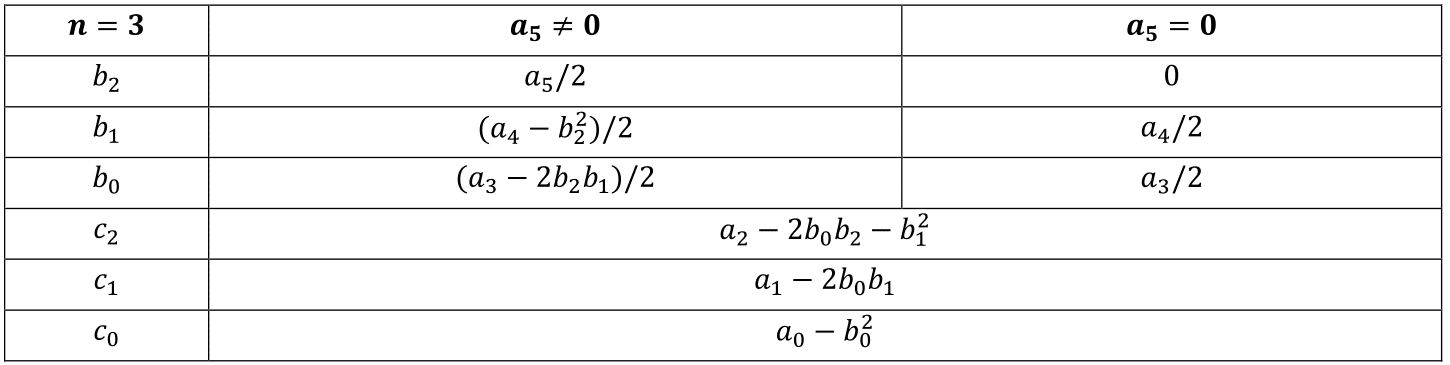}
\end{figure}

\begin{figure}[H]
\centering
\caption{Completing the square for degree 8}
\label{degree8}
\includegraphics[scale=0.4]{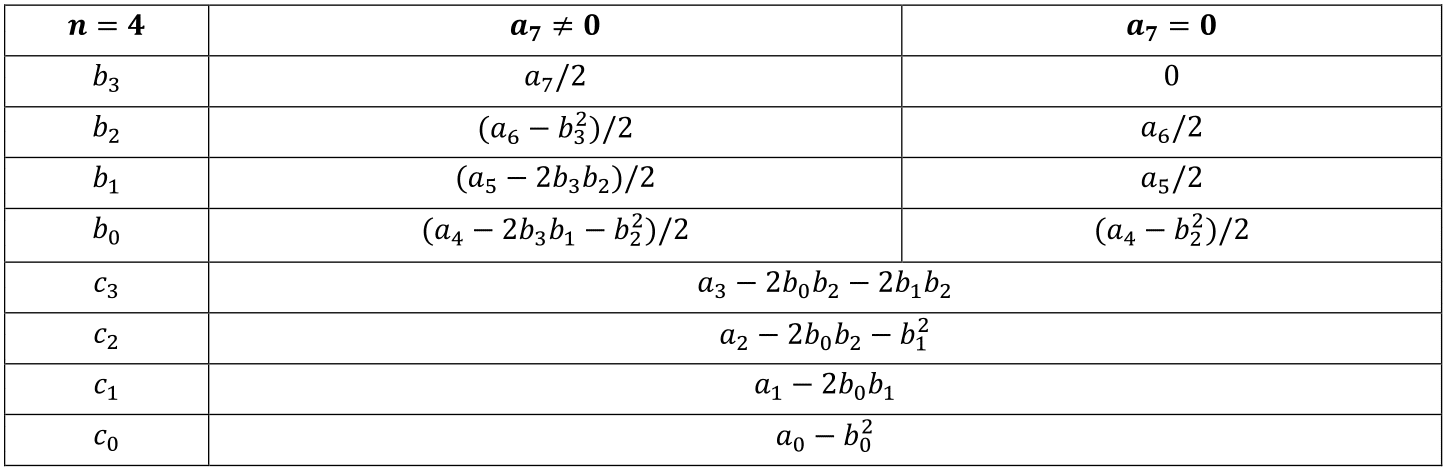}
\end{figure}

\begin{figure}[H]
\centering
\caption{Completing the square for degree 10}
\label{degree10}
\includegraphics[scale=0.4]{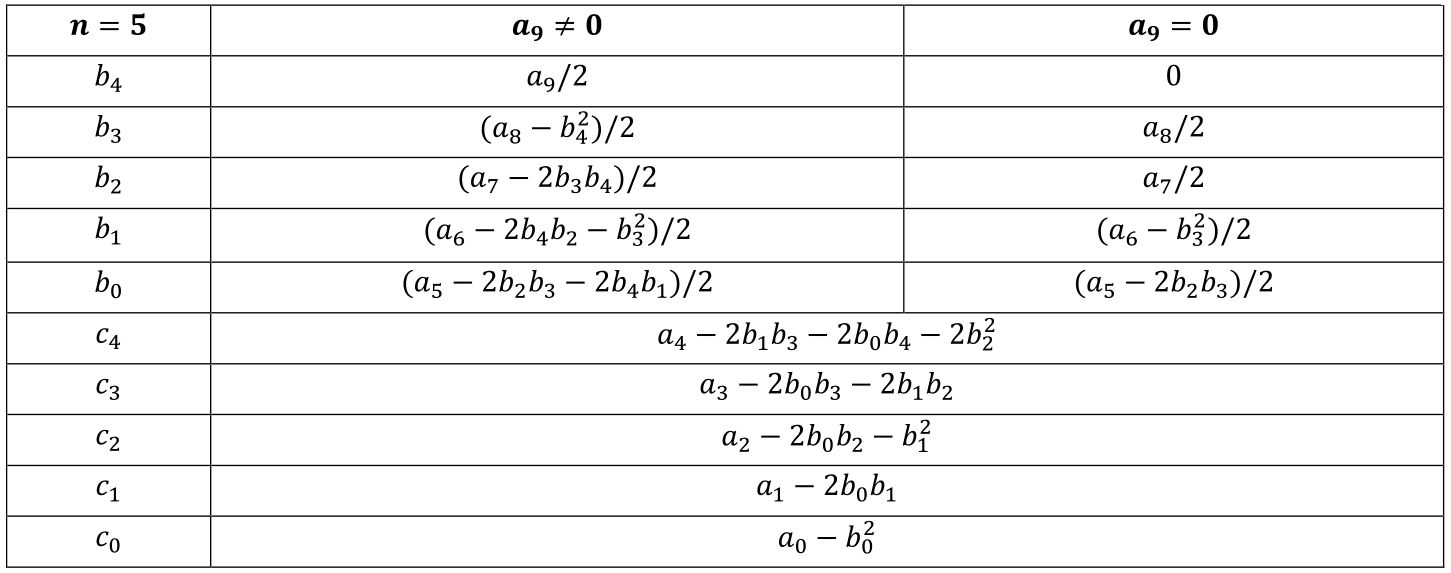}
\end{figure}

\begin{figure}[H]
\centering
\caption{Completing the square for degree 12}
\label{degree12}
\includegraphics[scale=0.4]{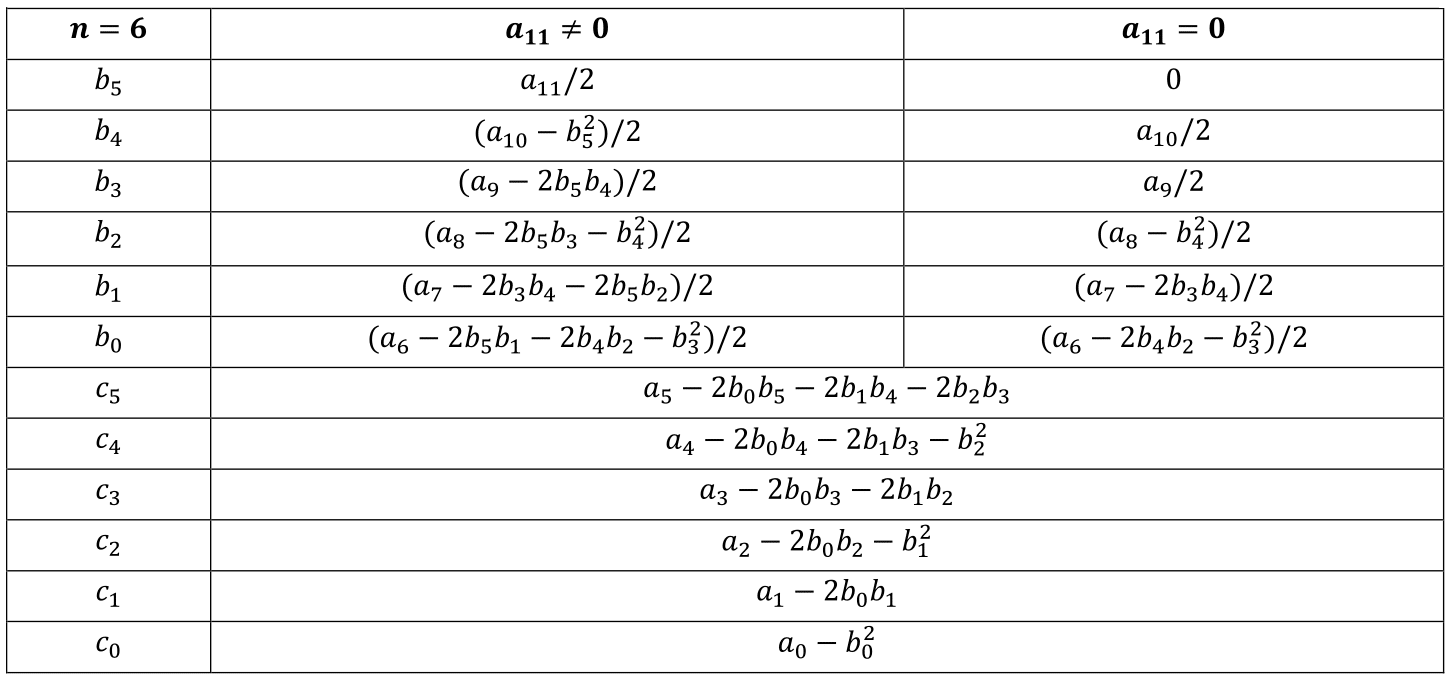}
\end{figure}

\begin{figure}[H]
\centering
\caption{Completing the square for degree 14}
\label{degree14}
\includegraphics[scale=0.4]{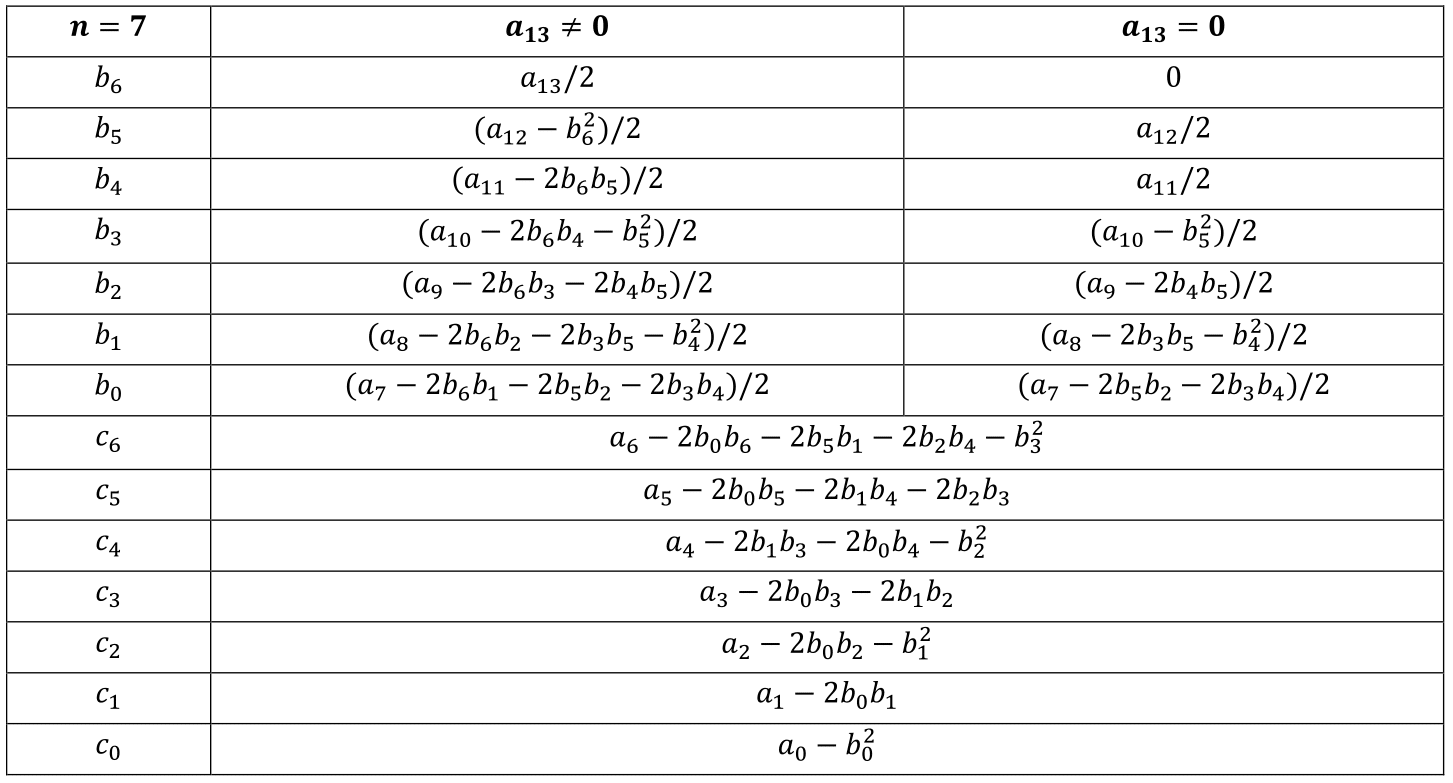}
\end{figure}

\end{document}